\documentclass[envcountsame,envcountchap]{article}
\usepackage{amssymb,amsmath}
\usepackage{amsfonts}
\usepackage{graphicx}
\usepackage[matrix,arrow,ps]{xy}

\usepackage{graphicx}
\usepackage{subfigure}
\usepackage[usenames]{color}
\usepackage{amssymb,amsmath,amsthm}
\usepackage{amsfonts}
\usepackage{graphicx}
\usepackage{type1cm}
\usepackage{eso-pic}
\usepackage{color}
\usepackage{makeidx}
\usepackage[section]{placeins}
\usepackage{float}

\newtheorem{theorem}{Theorem}[section]

\begin{document}

\newcommand{\Ito}{It\^{o} }
\newcommand{\bz}{{\bf z}}
\newcommand{\vv}{{\bf v}}
\newcommand{\ww}{{\bf w}}
\newcommand{\yy}{{\bf y}}
\newcommand{\xx}{{\bf x}}
\newcommand{\nn}{{\bf n}}
\newcommand{\uu}{{\bf u}}
\newcommand{\mm}{{\bf m}}
\newcommand{\qq}{{\bf q}}
\newcommand{\aba}{{\bf a}}
\newcommand{\pp}{{\bf p}}
\newcommand{\OO}{\mathbb{O}}
\newcommand{\II}{\mathbb{I}}
\newcommand{\IR}{\mathbb{R}}
\newcommand{\IC}{\mathbb{C}}
\newcommand{\IB}{\mathbb{B}}
\newcommand{\IZ}{\mathbb{Z}}
\newcommand{\half}{\frac{1}{2}}
\newcommand{\halff}{1/2}
\newcommand{\bea}{\begin{eqnarray*}}
\newcommand{\eea}{\end{eqnarray*}}
\newcommand{\beaq}{\begin{eqnarray}}
\newcommand{\eeaq}{\end{eqnarray}}
\newcommand{\bfalpha}{\mbox{\boldmath $\alpha$ \unboldmath} \hskip -0.05 true in}
\newcommand{\bfgamma}{\mbox{\boldmath $\gamma$ \unboldmath} \hskip -0.05 true in}
\newcommand{\bfmu}{\mbox{\boldmath $\mu$ \unboldmath} \hskip -0.05 true in}
\newcommand{\bfnu}{\mbox{\boldmath $\nu$ \unboldmath} \hskip -0.05 true in}
\newcommand{\bfxi}{\mbox{\boldmath $\xi$ \unboldmath} \hskip -0.05 true in}
\newcommand{\bfphi}{\mbox{\boldmath $\phi$ \unboldmath} \hskip -0.05 true in}
\newcommand{\bfpi}{\mbox{\boldmath $\pi$ \unboldmath} \hskip -0.05 true in}
\newcommand{\beq}{\begin{equation}}
\newcommand{\eeq}{\end{equation}}
\newcommand{\bfomega}{\mbox{\boldmath $\omega$ \unboldmath} \hskip -0.05 true in}
\newcommand{\bfeta}{\mbox{\boldmath $\eta$ \unboldmath} \hskip -0.05 true in}
\newcommand{\bfepsilon}{\mbox{\boldmath $\epsilon$ \unboldmath} \hskip -0.05 true in}
\newcommand{\bftheta}{\mbox{\boldmath $\theta$ \unboldmath} \hskip -0.05 true in}
\newcommand{\om}{\omega}
\newcommand{\Om}{\Omega}
\newcommand{\bom}{\bfomega}
\newcommand{\III}{\rm III}
\def\tab{ {\hskip 0.15 true in} }
\def\vtab{ {\vskip 0.1 true in} }
 \def\htab{ {\hskip 0.1 true in} }
 \def\ntab{ {\hskip -0.1 true in} }
 \def\vtabb{ {\vskip 0.0 true in} }

\begin{center}
{\bf \LARGE Rate of Entropy Production in Stochastic Mechanical Systems}
\end{center}

\begin{center}
Gregory S. Chirikjian \\
National University of Singapore \\
November 27, 2021
\end{center}

\begin{abstract}
Entropy production in stochastic mechanical systems 
is examined here with strict bounds on its rate. Stochastic mechanical  systems include pure diffusions in Euclidean space or on Lie groups, as well as systems evolving on phase space  for which the fluctuation-dissipation theorem applies, i.e.,  return-to-equilibrium processes. Two separate ways for ensembles of such mechanical systems forced by noise to reach equilibrium are examined here. First, a restorative potential and damping can be applied, leading to a classical return-to-equilibrium process wherein energy taken out by damping can balance the energy going in from the noise. Second, the process evolves on a compact configuration space (such as random walks on spheres, torsion angles in chain molecules, and rotational Brownian motion) lead to long-time solutions that are constant over the configuration space, regardless of whether or not damping and random forcing balance. This is a kind of potential-free equilibrium distribution resulting from topological constraints. Inertial and noninertial (kinematic) systems are considered. These systems can
consist of unconstrained particles or more complex systems with constraints, such as rigid-bodies or linkages.
These more complicated systems evolve on Lie groups and model phenomena such as rotational Brownian motion and nonholonomic robotic systems. In all cases, it is shown that the rate of entropy production is closely related to the appropriate concept of Fisher information matrix of the probability density defined by the Fokker-Planck equation. Classical results from information theory are then repurposed to provide computable bounds on the rate of entropy production in stochastic mechanical systems.
\end{abstract}

\section{Introduction}

According to the second law of thermodynamics, at the macroscopic scale entropy is nondecreasing.
Statistical mechanical arguments for this result were established by Boltzmann's classical work for
ideal gasses involving deterministic collision models between particles. A completely different approach
is the stochastic mechanical model in which each element in the larger system is being forced by random Brownian motion and viscous damping. Such models are not limited to particle systems, and can either be modelled as having inertia or can be purely kinematic. They can be complex molecules or even systems with nonholonomic kinematic constraints such as robots.

Stochastic mechanical models (i.e., stochastic differential equations and associated Fokker-Planck equations) are examined here with the goal of establishing strict bounds on the rate of entropy production. Within this context, each representative in an ensemble of a stochastic mechanical systems has a family of probability density functions (indexed by time) from which entropy is computed.

A probability density function (pdf) on a measure space $(X,\mu)$ which is indexed by time
is a function $f:X\times \IR_{\geq 0} \,\longrightarrow\, \IR_{\geq 0}$ such that
$$ \int_X f(\xx,t) \, d\mu(\xx) \,=\,1\,. $$
The entropy of $f$ at each value of time is defined as
\beq
S(t) \,\doteq\, -\int_X f(\xx,t) \log f(\xx,t) \, d\mu(\xx) \,.
\eeq
The choice of base of the logarithm amounts to a choice of measurement units for $S$, as a different base will result in a different scaling of $S(t)$. Throughout this article, $\log \doteq \log_e = \ln$.

When considering a statistical mechanical system, the entropy of a time-varying probability density on phase space is defined as\footnote{The inclusion of the Boltzmann constantt $k_B$ here is for consistency with the statistical mechanics literature.}
\beq
S_B(t) \,\doteq\, -k_B \int_{\qq} \int_{\pp} f(\pp,\qq,t) \log f(\pp,\qq,t) \, d\pp \, d\qq
\eeq
where $\qq \in Q \subseteq \IR^n$ is a set of coordinates, and $\pp \in \IR^n$ are the corresponding conjugate momenta. The Lebesgue measure on the $2n$-dimensional phase space, $d\pp \, d\qq = dq_1 ... dq_n dp_1 ... dp_n$, is invariant under coordinate changes, though such changes do affect the bounds of integration $Q$, unless $Q = \IR^n$.

The configurational probability density is the marginal distribution
\beq
f(\qq,t) \,\doteq\, \frac{1}{\sqrt{M(\qq)}} \int_{\pp} f(\pp,\qq,t) \, d\pp
\label{fqdef}
\eeq
which is defined in this way so that
$$ \int_{\qq \in Q} f(\qq,t) \,\sqrt{M(\qq)}\, d\qq \,=\, 1 \,. $$
The corresponding configurational entropy is
\beq
S_Q(t) \,\doteq\, -\int_{\qq \in Q} f(\qq,t) \log f(\qq,t) \,\sqrt{M(\qq)}\, d\qq
\label{sqdef}
\eeq
where $M(\qq)$ is the metric tensor for the configuration space, which is the mass matrix in the case of a mechanical system with inertia. In the case when $\qq$ parameterizes a unimodular Lie group, $G$, such as the group of rotations or the Euclidean motion group, then the entropy of a time-indexed pdf $f:G\times\IR_{\geq 0}\,\rightarrow\,\IR_{\geq 0}$ is
\beq
S_G(t) \,\doteq\, -\int_{G} f(g,t) \log f(g,t) \,dg
\label{Gent}
\eeq
where the Haar measure $dg$ takes on different appearances under changes of parametrization but the value
of the integral is independent of parametrization and it
is invariant under shifts of the form $g \rightarrow g_0 g$ and $g \rightarrow g g_0$ for any fixed $g_0 \in G$.

The main topic of this article is to study the rate of entropy increase in stochastically forced mechanical systems. That is, for any of the entropies defined above, to calculate or bound $\dot{S}$. In particular, by simply
moving the time derivative inside the integral and using the product rule from Calculus,
\beq
\dot{S}(t) \,=\, - \int_{X} \left\{\frac{\partial f}{\partial t} \log f + \frac{\partial f}{\partial t}
\right\} \,d\mu(\xx) .
\label{sdotdef}
\eeq
But since $f(\xx,t)$ is a probability density function at each value of time, whose integral is equal to $1$,
$$ \int_{X} \frac{\partial f}{\partial t} \,d\mu(\xx) =
\frac{d}{dt} \int_{X} f(\xx,t) \,d\mu(\xx) = 0, $$
and so the second term in the above braces integrates to zero. Consequently,
\beq
\dot{S}(t) \,=\, - \int_{X} \frac{\partial f}{\partial t} \log f \,d\mu(\xx) .
\label{sdotdef111}
\eeq


It is also possible to bound the value of entropy itself.
The function $\Phi(x) = - \log x$ is a convex function. Consequently,
Jensen's inequality gives
$$ S \,=\, \int_X f(\xx) \Phi(f(\xx)) d\mu(\xx) \,\geq\, \Phi \left(\|f\|^2\right) $$
or
\beq
S \,\geq\, -\log \left(\|f\|^2\right)
\label{entmag1}
\eeq
where
$$ \|f\|^2 \,\doteq\, \int_X |f(\xx)|^2 d\mu(\xx) \,. $$
As a consequence of Parseval's inequality, (\ref{entmag1}) can then be stated in Fourier space.
This is true not only for the case of Euclidean space, but for wide classes of unimodular Lie groups
\cite{harmonic}.

Moreover, $S(t)$ can be bounded from above in some contexts. For example, on Euclidean spaces it is well known that Gaussian distributions have maximum entropy over all pdfs with a given mean and covariance.
Therefore, if $f(\xx,t)$ is
an arbitrary time-evolving pdf on $\IR^d$ with mean $\bfmu_f(t)$ and covariance $\Sigma_f(t)$, and if
$\rho_{\mu, \Sigma}(\xx)$ denotes a Gaussian distribution with mean $\bfmu$ and covariance $\Sigma$, then
\beq
S_f(t) \,\leq\, S_{\rho_{\mu_f, \Sigma_f}}(t) \,.
\label{maxent}
\eeq
On a compact space such as the circle or rotation group, the uniform distribution has the absolute maximum entropy of all distributions on those spaces.


%
%


\section{Rate of Entropy Production for Stochastic Processes on Euclidean Space}

The stochastic mechanical models addressed here are described as stochastic differential equations
forced by Brownian noise (i.e., Gaussian white noise, or equivalently, increments of Wiener processes).

\subsection{Review of Stochastic Differential Equations}

Stochastic differential equations (SDEs) are
forced by Gaussian white noises $dw_i$, which are increments of uncorrelated unit-strength Weiner processes.
These define Brownian motion processes. That is, each $dw_i(t)$ can be viewed as an independent random draw from a one-dimensional Gaussian distribution with zero mean and unit variance. Let $d$ denote the dimension of a Euclidean space on which the random process evolves. For systems with inertia, $d=2n$ is phase space, and for noninertial systems $d=n$. The
independent uncorrelated unit strength white noises $dw_i$ form the components of a vector $d{\bf w} \in \IR^d$.
For example, given the stochastic differential equation on $\IR^d$
$$ d\xx = B \,d{\bf w} $$
where $B \in \IR^{d\times d}$ is a constant full-rank matrix, the distribution $f(\xx,t)$ describing  the ensemble of an infinite number of
trajectories will satisfy
$$ \frac{\partial f}{\partial t} \,=\, \frac{1}{2} \sum_{i,j=1}^{d} D_{ij}
\frac{\partial^2 f}{\partial x_i \partial x_j}\,. $$
where $D = [D_{ij}] = BB^T$.
The solution to this equation subject to initial conditions $f(\xx,0) = \delta(\xx - {\bf 0})$ is
the time-varying Gaussian distribution
\beq
f(\xx,t) \,=\, \frac{1}{(2\pi)^{d/2} |Dt|^{1/2}} \exp(-\frac{1}{2} \xx^T (Dt)^{-1} \xx)
\label{Gaussian}
\eeq
where $|A|$ denotes the determinant of a square matrix $A$.
The entropy for this can be computed in closed form as \cite{Shannon}
\beq
S(t) \,=\, \log\left\{(2\pi e)^{d/2} |\Sigma(t)|^{1/2} \right\} \,.
\label{gaussianentropy}
\eeq
where $\Sigma(t) = Dt$ in this context. More generally (\ref{gaussianentropy}) can be used as the upper bound
in (\ref{maxent}) with $\Sigma = \Sigma_f$ since entropy is independent of the mean.

As an example of (\ref{entmag1}), applying it to a Gaussian distribution on $\IR^d$ gives
$$ S(t) \,\geq\, \log\left\{(4\pi)^{d/2} |Dt|^{1/2}\right\} . $$
Comparing with the exact expression in (\ref{gaussianentropy}) verifies this since $4\pi < 2\pi e$ and $\log(x)$ is a monotonically increasing function.

There are two major kinds of stochastic differential equations (SDEs), \Ito and Stratonovich. Both are
forced by Gaussian white noises $dw_i$. In the simple example above, \Ito and Stratonovich interpretations
lead to the same result, but in more complex cases where the coupling matrix $B$ is configuration-dependent,
the two intepretations will differ. A brief review of the main features of these two different interpretations of
SDEs is given here based on the longer exposition in \cite{stochastic}.

Historically, the \Ito interpretation came first. If
\begin{equation}
dx_i(t) = a_i(x_1(t),...,x_d(t),t) dt + \sum_{j=1}^{m}
B_{ij}(x_1(t),...,x_d(t),t) \, dw_j(t) \htab {\rm for} \htab i=1,...,d
\label{stratstoch1}
\end{equation}
is an \Ito SDE describing a random process on $\IR^d$,
where now $B \in \IR^{d\times m}$,
then the corresponding Fokker-Planck equation governing
the probability density of the ensemble of states, $f(\xx, t)$, is \cite{stochastic}
\begin{equation}
\frac{\partial f({\bf x},t)}{\partial t} =
- \sum_{i=1}^{d} \frac{\partial}{\partial x_i}\left(a_i({ \bf x},t) f({ \bf x},t)\right)
+ \half \sum_{i,j=1}^{d}
\frac{\partial^2}{\partial x_i \partial x_j}
\left(\sum_{k=1}^{m} B_{ik}({\bf x},t) B_{kj}^{T}({\bf x},t) f({ \bf x},t)\right),
\label{fokplan2}
\end{equation}
\Ito SDEs are popular in mathematical statistics contexts, measurement and filtering theory, and in finance, because of the ease with which expectations can be taken, and the associated martingale properties. In engineering contexts when modeling physical processes the Stratonovich interpretation of SDEs described below is more popular because standard rules of
Calculus can be used. For this reason the Stratonovich interpretation is popular in differential geometric contexts since moving between coordinate patches involves Calculus operations.
A Stratonovich SDE describing the exact same random process as the \Ito SDE given above is written as
\begin{equation}
dx_i(t) = a_i^s(x_1(t),...,x_d(t),t) dt + \sum_{j=1}^{m}
B_{ij} (x_1(t),...,x_d(t),t) \,\circledS \, dw_j(t) \htab {\rm for} \htab i=1,...,d
\label{stratstoch1a}
\end{equation}
where
\beq
a_i(\xx,t) = a_i^s(\xx,t) + \half \sum_{j=1}^{m} \sum_{k=1}^{d}
\frac{\partial B_{ij} }{\partial x_k} B_{kj} \,.
\label{conveqwufe}
\eeq
This illustrates that there is very simple way to interconvert between \Ito and Stratonovich by either adding or subtracting the last term in (\ref{conveqwufe}).

The Stratonovich form of the Fokker-Planck equation (FPE) is written as \cite{stochastic}
\beq
\frac{\partial f}{\partial t} = - \sum_{i=1}^d \frac{\partial}{\partial x_i} \left(a_i^s f \right)
+ \half \sum_{i,j=1}^{d} \frac{\partial}{\partial x_i} \left[\sum_{k=1}^{m} B_{ik}   \frac{\partial}{\partial x_j} (B_{jk}  f) \right]
\label{fpstrateqgqwow}
\eeq
When $B$ is independent of the configuration variable $\xx$, \Ito snd Stratonovich versions of the FPE are always the same, as can be seen from (\ref{conveqwufe}) where the discrepancy between the two versions of drift term vanishes as the partial derivatives of $B$ with respect to ${\bf x}$ vanish. This is a sufficient condition for \Ito and Stratonovich SDEs to yield the same FPE, but is not a necessary for this to be the case, as will be seen in later examples. Similar equations hold for processes evolving on manifolds and Lie groups, as will be discussed later in the paper.

The main topic addressed here is the rate at which $S(t)$ changes. This can be observed by substituting in the
solution of the Fokker-Planck equations into the definition of $\dot{S}$ in (\ref{sdotdef}).

\subsection{Rate Of Entropy Production}

Using the Stratonovich form of the FPE in (\ref{fpstrateqgqwow}), we arrive at the following theorem
\begin{theorem}
The rate of entropy production for $f(\xx,t)$ governed by (\ref{fpstrateqgqwow}) evolving freely on Euclidean space is positive and bounded from below by
\beq
\dot{S} \,\geq\, \frac{1}{2} \int_{\IR^d} {\rm trace}\left[\frac{(B(\xx,t)^T \nabla f) (B(\xx,t)^T \nabla f)^T}{f}\right] \,d\xx
\label{entrate1}
\eeq
when
\beq
\frac{\partial}{\partial x_i}\left(a_i^s - \sum_{k,j} B_{ik}  \frac{\partial B_{jk} }{\partial x_j}\right) \geq 0
\label{condf332222dg}
\eeq
with equality holding when
\beq
a_i^s - \sum_{k,j} B_{ik}  \frac{\partial B_{jk} }{\partial x_j} = c_i
\label{constdrift}
\eeq
is constant for all values of $i$. In the case when $B$ is a constant and
$D \doteq BB^T$, then
\beq
\dot{S} \,=\, \frac{1}{2}  {\rm trace}\left[D F(t) \right]\,.
\label{entrate2}
\eeq
where
$$ F(t) =  \int_{\IR^d} \frac{(\nabla f) (\nabla f)^T}{f} \,d\xx $$
is the Fisher information matrix of $f(\xx,t)$.
\end{theorem}
\begin{proof}
The $\partial f/\partial t$ term in
$$ \dot{S} \,=\, - \int_{\IR^d} \frac{\partial f}{\partial t} \log f \,d\xx  $$
can be expressed in terms of spatial derivatives by substituting in the FPE. This is written as two terms using integration by parts with the surface terms at infinity vanishing due to the fact that the process evolves freely and the pdf must decay to zero at infinity. First,
$$ - \int_{\IR^d} \left\{ - \sum_{i=1}^d \frac{\partial}{\partial x_i} \left(a_i^s f \right) \right\} \log f \,d\xx
\,=\, - \int_{\IR^d} \sum_{i=1}^d a_i^s \frac{\partial f}{\partial x_i} \,d\xx \,. $$
Second,
$$ - \int_{\IR^d} \left\{
\half \sum_{i,j=1}^{d} \frac{\partial}{\partial x_i} \left[\sum_{k=1}^{m} B_{ik}   \frac{\partial}{\partial x_j} (B_{jk}  f) \right] \right\} \log f \,d\xx
\,=\, \frac{1}{2} \int_{\IR^d} \frac{1}{f} \sum_{i,j=1}^{d} \left[\sum_{k=1}^{m} B_{ik}   \frac{\partial}{\partial x_j} (B_{jk}  f) \right] \frac{\partial f}{\partial x_i}
\,d\xx \,. $$
Expanding
$$ \frac{\partial}{\partial x_j} (B_{jk}  f) = \frac{\partial B_{jk} }{\partial x_j} f +
B_{jk}  \frac{\partial f}{\partial x_j} $$
and recollecting terms gives
$$ \dot{S} \,=\, \int_{\IR^d} \sum_{i=1}^d \left\{
-a_i^s + \sum_{i,j=1}^{d} B_{ik}  \frac{\partial B_{jk} }{\partial x_j}\right\} \frac{\partial f}{\partial x_i} d\xx \,+\,
\half \int_{\IR^d} \frac{1}{f} \sum_{k=1}^{m}
\left(\sum_{i=1}^{d}  B_{ik}  \frac{\partial f}{\partial x_i}\right)
\left(\sum_{j=1}^{d} B_{jk}  \frac{\partial f}{\partial x_j} \right)
d\xx \,.
$$
The second integral is always nonnegative, as it is a positive semi-definite quadratic form and can be written as
$$ \frac{1}{2} {\rm trace}\left[\int_{\IR^d}\frac{(B(\xx,t)^T \nabla f) (B(\xx,t)^T \nabla f)^T}{f} d\xx\right]\,. $$
The first term can have either sign. However, if
$$ a_i^s - \sum_{k,j} B_{ik}  \frac{\partial B_{jk} }{\partial x_j} = c_i $$
is constant, then the first integral will vanish, and if by integration by parts the derivative in the first term is transferred over, the condition in the statement of the theorem will result since $f \geq 0$.
\end{proof}

When (\ref{condf332222dg}) holds it is clear that a looser lower bound than (\ref{entrate1}) akin to
(\ref{entrate2}) can be obtained as
\beq
\dot{S} \,\geq\, \frac{1}{2}  {\rm trace}\left[D_0(t) F(t) \right]
\label{d0def}
\eeq
by constructing positive definite matrix $D_0(t) = D_0^T(t)$ such that the following matrix inequality
$$ D_0(t) \,\leq\, B(\xx,t)B^T(\xx,t) \,\,\,\forall\,\,\xx \in \mathbb{R}^d $$
is satisfied. The reason why (\ref{d0def}) then holds is because the trace is linear and the trace of the product of positive semi-definite matrices is nonnegative, and both $F$ and $BB^T - D_0$ are positive semi-definite.

The condition (\ref{constdrift}) and result (\ref{entrate2}) are for diffusion processes with constant diffusion tensor and drift. Given initial conditions
$f(\xx,0) = f_0(\xx)$, the solution $f(\xx,t)$ in this case will be of the form
\beq
f(\xx,t) \,=\, (f_0 * \rho_{t{\bf a}, tBB^T})(\xx)
\label{convinit}
\eeq
where $\rho_{\bfmu, \Sigma}(\xx)$ is a multivariate Gaussian with mean $\bfmu$ and covariance $\Sigma$.
The convolution of any two functions $f_1,f_2 \in (L^1 \cap L^2)(\IR^d)$ is defined as
\beq
(f_1 * f_2)(\xx) \,\doteq\, \int_{\IR^d} f_1(\yy) f_2(\xx - \yy) \, d\yy\,.
\label{convrd}
\eeq

Fisher information plays an important part in probability theory \cite{Fisherref,Kullbackref} and
its connections to physics also have been recognized \cite{st_jaynes,st_frieden}. For a recent review
of its properties see \cite{Zegers}.

Several inequalities from Information Theory can then be used to bound both entropy and entropy rate by quantities that are easily computable. For example, it is known that \cite{stam,blachman,dembo}
\beq
\frac{1}{{\rm tr}[F(f_1*f_2) P]} \,\geq\,
\frac{1}{{\rm tr}[F(f_1) P]} + \frac{1}{{\rm tr}[F(f_2) P]}
\label{fishrffftight}
\eeq
where $P$ is any real positive definite symmetric matrix with the same dimensions as $F$.
When $P=D$ this can then be used to give a lower bound on ${\rm tr}[FD]$, and hence on $\dot{S}$.
Moreover, one reason for the significance of (\ref{fishrffftight}) in information theory is that
it provides a path for proving the Entropy Power Inequality \cite{stam,blachman} described below.

The entropy power of a pdf $f(\xx)$ on $\IR^d$ was defined by Shannon
as \cite{Shannon,Cover}
$$ N(f) \,\doteq\, \frac{\exp(2 S(f)/d)}{2\pi e} $$
where $S(f)$ denotes the entropy of $f$.
The entropy power inequality is
\beq
N(f_1*f_2) \,\geq\, N(f_1) + N(f_2) \,.
\label{entpowine932}
\eeq
Since the logarithm is a strictly increasing function, this provides a lower bound on $S(f_1*f_2)$ and hence
can be used to bound the entropy of $f(\xx,t)$ of the form in (\ref{convinit}). In Section \ref{boundingcovsec}
lower bounds on $\dot{S}$ will be derived.

It should be noted that (\ref{fishrffftight}) and (\ref{entpowine932}) only apply for convolution on Euclidean spaces, and do not even apply for diffusion processes on the circle. In contrast, other bounds on non-Euclidean spaces and for processes that are not necessarily homogenoues diffusions are presented later in this paper.

\subsection{Examples}

\subsubsection{Brownian Motion in Euclidean Spaces}

Brownian motion in $d$-dimensional Euclidean space with Dirac delta initial conditions
was already reviewed, with pdf (\ref{Gaussian}) and
entropy (\ref{gaussianentropy}). From this, the rate of entropy production can be computed explicitly as
$$ \dot{S} \,=\,
\frac{(d/2)(2\pi e)^{d/2} |D|^{1/2} t^{d/2 -1} }{(2\pi e)^{d/2} |D|^{1/2} t^{d/2}} = \frac{d}{2t} \,. $$

The version of the Fisher information matrix in the theorem for a Gaussian is the inverse of the covariance, and hence
$$ F \,=\, (Dt)^{-1} = t^{-1} D^{-1}  \,. $$
Consequently (\ref{entrate2}) gives
$$ \dot{S} \,=\, \frac{1}{2} {\rm trace} (t^{-1} \II) = \frac{d}{2t} \,. $$

\subsubsection{Brownian Motion on the Torus/Circle}

The stochatic differential equation
$$ d\theta = \sqrt{D} dw $$
with constant scalar $D$
describing Brownian motion on the unit circle has an associated Fokker-Planck equation
$$ \frac{\partial f}{\partial t} = \half D \frac{\partial^2 f}{\partial \theta^2} $$
which is the same as the case on the line. However, the boundary condition $\theta(\pi) = \theta(-\pi)$
is imposed rather than free boundary conditions.

Let
$$ \rho(x,t) \,\doteq\, \frac{1}{\sqrt{2\pi Dt}} e^{-\,\frac{x^2}{2Dt}} \,, $$
which is the solution to the heat equation on the real line subject to initial conditions
$\rho(x,0) = \delta(x)$.

The solution to the heat equation on the circle subject to initial condition $f(\theta,0) =\delta(\theta)$
is then \cite{stochastic,harmonic}
\beq
f(\theta,t) \,=\, \sum_{k=-\infty}^{\infty} \rho(\theta - 2\pi k; t) = \frac{1}{2\pi} + \frac{1}{\pi}  \sum_{n=1}^{\infty}
e^{-Dt n^2/2} \cos n\theta
\label{heatcirc8939}
\eeq
The above two equalities represent two very different ways of describing the solution. In the first equality,
the Gaussian solution presented in the previous section is wrapped (or folded) around the circle.
The second solution is a Fourier series solution. The first is efficient when $kt$ is small. In such cases only $k=0$ may be sufficient as an approximation. When $kt$ is large, truncating $n=1$ in the Fourier expansion may be
sufficient.

Statistical quantities such as the mean and variance can be computed in closed form as
$$ \mu(t) \,=\, \int_{-\pi}^{\pi} \theta f(\theta,t) d\theta \,=\,  0 $$
and
$$ \sigma^{2}(t) \,=\, \int_{-\pi}^{\pi} \theta^2 f(\theta,t) d\theta \,=\,  \frac{\pi^2}{3} \,+\, 4 \sum_{n=1}^{\infty} \frac{(-1)^n}{n^2} e^{-Dt n^2/2} \,. $$

In both expressions in (\ref{heatcirc8939}) summations are present which makes the exact analytical computation of logarithms, and hence entropy, problematic.  However, Since the function $\Phi(x) = -\log x$ is a monotonically decreasing function,
then when $a,b >0$ we have $\Phi(a+b) < \Phi(a)$. Consequently, since $\rho >0$ everywhere,
\beq
S(t) \,\leq\, -\int_{-\pi}^{\pi} f(\theta,t) \log \rho(\theta,t) d\theta
= \frac{1}{2} \log(2\pi Dt) + \frac{\sigma^2(t)}{2Dt} \,.
\label{lbentcirc}
\eeq
In the extreme case when the distribution reaches equilibrium $f_{\infty}(\theta) = 1/(2\pi)$, this is the maximum entropy possible, and hence
\beq
S(t) \,\leq\,  S_{\infty} = \log(2\pi) \,.
\label{ubentcirc}
\eeq

It should be noted that all calculations here are done relative to the measure $d\theta$.
For a compact space like the circle, it is common to normalize the measure so that
$$ V = \int_{V} \,1\, dV = 1. $$
In doing so for the heat kernel on the circle
this would involve redefining $dV = d\theta/2\pi$ and $f'(\theta,t) \doteq 2\pi f(\theta,t)$.
This has no effect on mean and covariance, but the value of entropy is shifted such that the entropy of
the uniform distribution is equal to zero and the entropy of all other distributions are negative. Rewriting (\ref{ubentcirc}) in a way that does
not depend on the choice of normalization of measure is
\beq
S(t) \,\leq\, \log V\left(\mathbb{S}^1\right) \,.
\label{xxdfentcirc}
\eeq
In other words, the value of entropy not only depends on the base of the logarithm, but also on the way that the
integration measure is scaled.
%

\subsubsection{Concentration of Species Transport in Inhomogeneous Compressible Flow} \label{compressible}

The concentration $c(x,t)$ of a species in inhomogeneous compressible flow can be modeled by the
equation \cite{fluids1,fluids2,fluids3}
\beq
\frac{\partial c}{\partial t} \,=\, \frac{\partial}{\partial x} \left[D(x,t) \frac{\partial c}{\partial x}\right] - \frac{\partial}{\partial x}\left(u(x,t) c\right)
\label{compeq}
\eeq
where $D(x,t) = D_0(1+\kappa_0 x)^2$ and $u(x,t) = u_0(1+\kappa_0 x)$. This is a FPE, and it is possible to work backwards to find the corresponding SDE.

We see immediately from the above pde that if $c(x,0)$ is normalized to be a probability density function,
then $c(x,t)$ will preserve this property.

Equation (\ref{compeq}) represents a one-dimensional example of what was presented in the theorem, and the rate of entropy increase
is
$$ \dot{S} \,=\, \int_{-\infty}^{\infty} \frac{1}{c} D(x,t) \left(\frac{\partial c}{\partial x}\right)^2 dx
+ u_0 \kappa_0 $$
since the drift term in the entropy rate computation simplifies as
$$ - \int_{-\infty}^{\infty} u \frac{\partial c}{\partial x} dx \,=\,
\int_{-\infty}^{\infty} c \frac{\partial u}{\partial x} dx = u_0 \kappa_0 \int_{-\infty}^{\infty} c(x,t) dx = u_0 \kappa_0\,. $$

\subsubsection{Homogeneous Transport in Couette Flow}

As another example from classical fluid mechanics, consider 2D homogeneous transport in Couette flow governed
by the equation \cite{sun,fluids4}
$$ \frac{\partial c}{\partial t} \,=\, D_0 \left[\frac{\partial^2 c}{\partial x^2}
+ \frac{\partial^2 c}{\partial y^2} \right] - \frac{U_0}{H} y \frac{\partial c}{\partial x} $$
where $y \in [0,H]$ and again $c(\xx,0)$ is normalized to be a probability density function, and hence
$c(\xx,t)$ retains this property. In this case the region over which the equation holds is an infinite slab with the concentration and its gradient vanishing at $y=0$ and $y=H$.

We then arrive at
$$ \dot{S} \,=\, D_0 \int_{0}^{H} \int_{-\infty}^{\infty} \frac{1}{c} \left[\left(\frac{\partial c}{\partial x}\right)^2 + \left(\frac{\partial c}{\partial y}\right)^2\right] dx dy
+ {\bfmu} \cdot {\bf e}_2 $$
where
$$ \bfmu \,=\,  \int_{0}^{H} \int_{-\infty}^{\infty} \xx \, c(\xx,t) d\xx $$
is the mean of $c(\xx,t)$.

The pdf for concentration in both of the above examples can be solved in closed form using the methods
in \cite{sun} in which inhomogeneous processes on Euclidean space are recast as homogeneous processes on an
appropriately chosen Lie group. But the purpose of these examples is to illustrate the relationship between
entropy rate and Fisher information. In the next section, a classical result of Information Theory is used
to bound the Fisher information with covariance. Covariance can be propagated without explicit knowledge of the
pdf, which will be demonstrated.

\section{Bounding Rate of Entropy Production with Covariance} \label{boundingcovsec}

The version of the Fisher information matrix that appears in entropy rate computations
in general is not easy to compute. An exception to this statement is when the pdf is a Gaussian
and
$$ F_{gauss} \,=\, \Sigma_{gauss}^{-1} \,. $$
For other exponential families closed form expressions are also possible.
But in general computing the time-varying Fisher information matrix for
a pdf satisfying a Fokker-Planck equation will not be easy to compute.

But it is possible to bound the Fisher information matrix with the covariance of the pdf, and
covariance can be propagated as an ordinary differential equation even when no explicit solution
for the time-varying pdf is known.

In this section the
Cram\'{e}r-Rao Bound \cite{cramerref,Rao}, a famous inequality in Euclidean statistics and information theory, is reviewed. Recently it has been extended to manifolds and Lie groups \cite{infotaxis,bonnabel1,bonnabel2,solo1,solo2,infomatr}. A second kind of inequality for compact
spaces such as tori, spheres, and rotation groups was introduced in \cite{wirtinger}. Then examples where
covariance is propagated directly from the FPE are given in which the aforementioned inequalities can be
put to use in bounding the rate of entropy generation. An outline of the general procedure is given here.

Let $\bfmu$ and $\Sigma$ denote the mean of a pdf $f(\xx;\bfmu, \Sigma)$. Then
$$ \bfmu \,\doteq\, \int_{\IR^n} \xx \, f(\xx;\bfmu,\Sigma) \, d\xx \,,$$
or equivalently
\beq
\int_{\IR^n} (\xx - \bfmu) \, f(\xx;\bfmu,\Sigma) \, d\xx \,=\, 0 \,,
\label{3j3j32}
\eeq
and
$$ \Sigma \,\doteq\, \int_{\IR^n} (\xx - \bfmu) (\xx - \bfmu)^T \, f(\xx;\bfmu,\Sigma) \, d\xx \,. $$

When presented with a FPE of the form
$$ \frac{\partial f}{\partial t} \,=\, {\cal D} f \,, $$
where ${\cal D}$ is as in the right hand side of (\ref{fpstrateqgqwow}),
an ordinary differential equation describing the evolution of $\bfmu(t)$ and $\Sigma(t)$ can be defined
as
$$ \dot{\bfmu} = \int_{\IR^n} \xx \, {\cal D} f \, d\xx $$
and
$$ \dot{\Sigma} = \int_{\IR^n} (\xx - \bfmu) (\xx - \bfmu)^T {\cal D} f \, d\xx $$
where integration by parts can be used in some cases to obtain explicit closed-form expressions for
the integrals on the right hand side of both equations. Once $\Sigma(t)$ is obtained, it can be used
to bound entropy from above using (\ref{gaussianentropy}) the rate of entropy production from below
using the results given in the next section.

\subsection{The Cram\'{e}r-Rao Bound}

The Cram\'{e}r-Rao Bound (or CRB) is a way to bound the covariance of an estimated statistical quantity \cite{cramerref,Rao}. Here it will not be used in its most general form. Here it will only be used in the unbiased estimation of the
mean of a probability density function on $\IR^n$. .

In the standard derivation of the CRB, as given in \cite{jljlopt,Cover} for the case of estimation of the mean,
the gradient of this expression with respect to $\bfmu$ is computed, where
where ${\partial f}/{\partial \bfmu}$ denotes the gradient as a column vector, and  ${\partial f}/{\partial \bfmu^T} \doteq [{\partial f}/{\partial \bfmu}]^T$.

Differentiation of both sides of  (\ref{3j3j32}) with respect to $\bfmu^T$ gives
$$ \frac{\partial}{\partial \bfmu^T} \int_{\IR^n} [{\bf x} - \bfmu] \,f\, d\xx =
\int_{\IR^n} [{\bf x} - \bfmu] \frac{\partial \,f\,}{\partial \bfmu^T} d\xx - \II = \OO\,. $$
Here the derivative is taken under the integral. The product rule
for differentiation then is used with the fact that $f$ is a pdf in $\xx$.
$\OO$ denotes the $m \times m$ zero matrix resulting from computing
${\partial {\bf 0}}/{\partial \bfmu^T}$.

The above equation can be written as \cite{cramerref,Cover,Rao}
\beq
\II = \int_{\IR^n} {\bf a}(\xx,\bfmu) {\bf b}^T(\xx,\bfmu) d\xx \in \IR^{p \times m}
\label{kbh443}
\eeq
where
$$ {\bf a}(\xx,\bfmu) = [\,f\,]^{\half} [{\bf x} - \bfmu] $$
and
$$ {\bf b}(\xx,\bfmu) = [\,f\,]^{-\half} \frac{\partial f}{\partial \bfmu} \,. $$
Using the fact that
$f(\xx;\bfmu,\Sigma) = f(\xx-\bfmu; {\bf 0},\Sigma)$
means that
$$ {\bf b}(\xx,\bfmu) = -[\,f\,]^{-\half} \frac{\partial f}{\partial \xx} \,. $$

Then it becomes clear that
\beq
 F \,=\, \int_{\IR^n} {\bf b}(\xx,\bfmu)  [{\bf b}(\xx,\bfmu)]^T d\xx
 \label{fdef0033}
 \eeq
and
\beq
\Sigma \,=\,
 \int_{\IR^n} {\bf a}(\xx,\bfmu) [{\bf a}(\xx,\bfmu)]^T d\xx.
 \label{cdef0033}
 \eeq

 Following the logic in \cite{jljlopt}, two arbitrary vectors are introduced: ${\bf v} \in \IR^{p}$ and ${\bf w} \in \IR^m$. Then (\ref{kbh443}) is multiplied on the left by ${\bf v}^T$
  and on the right by ${\bf w}$ to give
 \beq
 {\bf v}^T \II \, {\bf w} = \int_{\IR^n} {\bf v}^T {\bf a}(\xx,\bfmu) {\bf b}^T(\xx,\bfmu) {\bf w} \, d\xx.
 \label{kkene}
 \eeq
Regrouping terms in the resulting expression, squaring, and using the Cauchy-Schwarz inequality, then gives
\cite{cramerref,Rao,jljlopt}
\bea
\left(\int_{\IR^n} {\bf v}^T ({\bf a} {\bf b}^T) {\bf w} \, d\xx \right)^2 &=&
\left(\int_{\IR^n} ({\bf v}^T {\bf a})({\bf b}^T {\bf w}) d\xx \right)^2 \\ \\ &\leq&
\left(\int_{\IR^n} ({\bf v}^T {\bf a})^2 d\xx \right)
\left(\int_{\IR^n} ({\bf w}^T {\bf b})^2 d\xx \right) \\ \\
&=& \left( \int_{\IR^n} {\bf v}^T {\bf a} {\bf a}^T {\bf v} \, d\xx\right)
\left(\int_{\IR^n} {\bf w}^T {\bf b}  {\bf b}^T {\bf w} \, d\xx \right).
\eea
From the equations (\ref{fdef0033}), (\ref{cdef0033})  and (\ref{kkene}), this can be written as
$$ ({\bf v}^T \II \,{\bf w})^2 \leq ({\bf v}^T \Sigma {\bf v}) ({\bf w}^T F {\bf w}). $$
Making the choice of ${\bf w} = F^{-1} {\bf v}$ yields
$$ \left({\bf v}^T \, F^{-1}  {\bf v} \right)^2 \leq ({\bf v}^T \Sigma {\bf v}) \left({\bf v}^T \, F^{-1} {\bf v} \right). $$
This simplifies to
\beq
{\bf v}^T \left(\Sigma - \, F^{-1} \right) {\bf v} \geq 0
\htab \htab {\rm for} \htab {\rm arbitrary} \htab \htab
{\bf v} \in \IR^n
\label{crll33l}
\eeq
Consequently, the term in parenthesis is a positive definite matrix, or as a matrix inequality \cite{cramerref,Cover,Rao,jljlopt}

\beq
\Sigma \geq  \, F^{-1} \,,
\label{cr3kn32}
\eeq
which is the famous Cram\'{e}r-Rao Bound (for the special case of an unbiased estimator of the mean).
This is equivalently
\beq
\Sigma^{-1} \leq  \, F \,.
\label{cr3kn32a}
\eeq
Then, for example, in all of the equations for entropy production in time-varying pdfs on Euclidean space
presented earlier, it is possible to bound from below using a cascade of inequalities such as
$$ {\rm tr}[D F] \,\geq\, \lambda_{min}(D) \, {\rm tr}[F] \,\geq\,  \lambda_{min}(D) \, {\rm tr}\left[\Sigma^{-1}\right]\,. $$

\subsection{An Example}

Returning to the example of species transport in a compressible 1D flow outlined in Section \ref{compressible},
this section illustrates how the entropy rate can be bounded from below using the CRB even when a closed-form solution for the pdf is not known.

From the FPE itself, it is possible to propagate the mean and covariance. Multiplying both sides of (\ref{compeq}) and integrating by parts gives the following ordinary differential equation (ODE) for the mean $\mu(t)$
$$ \dot{\mu} \,=\, (2d_0\kappa_0 + u_0) (1 + \kappa_0 \mu) $$
subject to initial conditions, $\mu(0) = \mu_0$. This ODF can be solved in closed form for $\mu(t)$. However,
even if it could not be, it could be solved by numerical integration, which is much easier than solving the FPE. Similarly, since
$$ \sigma^2 =  \int_{-\infty}^{\infty} (x-\mu)^2 c(x,t) dx
= -\mu^2 + \int_{-\infty}^{\infty} x^2 c(x,t) dx \,, $$
multiplying (\ref{compeq}) by $x^2$ and integrating by parts
gives a way to propagate the covariance with an ODE of the form
$$ \frac{d}{dt}(\sigma^2) = F(\mu,\sigma^2) $$
which can be solved either analytically or numerically subject to initial conditions
$\sigma^2(0) = \sigma_0^2$.

It is worth noting that even in cases where such propagation of moments by FPE is not possible
(for example, when higher moments creep into the equations so that there is not form closure),
it is still possible to numerically generate a large ensemble of sample paths from the SDE corresponding to the FPE and compute variance (or covariance in multi-dof systems). Covariance estimation is much more stable than pdf
estimation, and so using the CRB as a lower bound is more reliable than directly attempting to compute
entropy, entropy rate, or Fisher information when the pdf is not known explicitly.

\section{Classical Statistical Mechanics as Stochastic Mechanics}

Classical Statistical Mechanics, as developed by Boltzmann and Gibbs, states that entropy increases.
For an introduction to phase space and equilibrium statistical mechanics see \cite{st_GibbsJW}.
Nonequilibrium statistical mechanics has been studied extensively over a long period of time starting with
Boltzmann and summarized in a number of books including
\cite{st_Prigoginefd, st_McLennan}. Important results continue to be developed in modern time, e.g., \cite{Jarzynski}. An alternative to the classical Boltzmann-Gibbs formulation is stochastic mechanics \cite{st_Kacfd,st_Bismut,st_nelson,nelson}. Here a Hamiltonian formulation of stochastic mechanics is used.

The Hamiltonian of a mechanical system is defined as the total system energy written in terms the conjugate momenta $\pp$ and generalized coordinates $\qq$:
\beq
 H(\pp,\qq) \doteq \half \pp^T M^{-1}(\qq)\, \pp \, + \, V(\qq) \,.
\label{hamlsll4432}
\eeq
Here $M(\qq)$ is the configuration-dependent mass matrix and
$${\bf p} \,\doteq\, M(\qq) \dot{\qq} \,. $$
The beauty of the Hamiltonian formulation is
that the volume in phase space (i.e., the joint $\pp$-$\qq$ space) is invariant under coordinate changes.

\subsection{Properties of Phase Space}

As is well known and explained in \cite{stochastic},
if ${\bf q}$ and ${\bf q}'$ are two different sets of coordinates, then
$$ T \,=\, \half \dot{{\bf q}}'^T M'({\bf q}') \dot{{\bf q}}' = \half \dot{{\bf q}}^T M({\bf q}) \dot{{\bf q}}. $$
Then with the Jaccobian relating rates of change,
$$ \dot{\bf q}' \,=\,  J({\bf q}) \, \dot{\bf q}, $$
it is clear that
$$ M({\bf q}) \,=\, J^T({\bf q}) M'({\bf q}'({\bf q})) J({\bf q}) \,. $$

From the above, and the definition of conjugate momenta,
\beq
{\bf p}' = J^{-T}({\bf q}) \, {\bf p}.
\label{jteq}
\eeq
Therefore, the two phase spaces have volume elements that are related as
$$  d{\bf p}' \,d{\bf q}' \,=\, \left|\begin{array}{cc}
J^{-T}({\bf q}) & \partial (J^{-T}({\bf q}) {\bf p})/\partial {\bf q} \\ \\
\mathbb{O} & J({\bf q})
\end{array}\right| \,d{\bf p} \,d{\bf q}. $$
The determinant of the upper-triangular block matrix in the above equation is equal to $1$, illustrating the invariance
\beq
d{\bf p} \,d{\bf q} \,=\, d{\bf p}' \,d{\bf q}'\,.
\label{invarphase}
\eeq
The key to this result is how $\pp$ transforms in (\ref{jteq}). A similar result holds in the Lie group
setting wherein the cotangent bundle of a Lie group can be endowed with an operation making it unimodular\footnote{A unimodular Lie group is one with an integration measure that is invariant under shifts from the left and the right.} even when the underlying group is not \cite{amitesh}. This is analogous to why (\ref{sqdef}) requires the metric tensor weighting and is coordinate dependent and (\ref{invarphase}) is not.

\subsection{Hamilton's Equations for a System Forced by External Noise and Damping}

Hamilton's equations of motion are
\beaq
\frac{d p_i}{dt} &\,=\,& - \frac{\partial H}{\partial q_i} + F_i \\
\frac{d q_i}{dt} &\,=\,& \quad \frac{\partial H}{\partial p_i} \,.
\label{hameqmoteq}
\eeaq
where $F_i$ are generalized external forces. In the case in which the mechanical system is forced by noise and viscous damping, then after multiplication by $dt$, these equations of motion become
\beaq
dp_i \,=\, -\half \pp^T \frac{\partial M}{\partial q_i} \pp \,dt \,-\, \frac{\partial V}{\partial {q}_i} \,dt\, -\,
{\bf e}_i^T CM^{-1} \pp \,dt \,+\, {\bf e}_i^T B d{\bf w}
\label{hffsgf1}
\eeaq
and
\beq
dq_i \,=\, {\bf e}_i^T M^{-1} \pp \,dt
\label{hffsgf2}
\eeq
where ${\bf e}_i$ is the $i^{th}$ natural unit basis vector. Note that the configuration-dependant
mass matrix $M = M(\qq)$, noise matrix $B = B(\qq)$, and damping $C = C(\qq)$ appear prominently in these equations.

Equations (\ref{hffsgf1}) and (\ref{hffsgf2}) can be written together as
\beq
\left(\begin{array}{c}
d\qq \\
d\pp \end{array} \right) = \left(\begin{array}{c}
\bfalpha(\pp,\qq) \\
\bfgamma(\pp,\qq) \end{array} \right) \,dt + \left(\begin{array}{cc}
\OO & \OO \\
\OO & B(\qq) \end{array} \right)  \left(\begin{array}{c}
d\ww' \\
d\ww \end{array} \right) \,.
\label{ksdknlwo33er}
\eeq
(Here $d\ww'$ multiplies zeros and hence is inconsequential).

The vector-valued function $\bfalpha$ and $\bfgamma$ are defined by their entries
\bea
\alpha_i &\doteq& {\bf e}_i^T M^{-1} \pp \\
\gamma_i &\doteq& -\half \pp^T \frac{\partial M}{\partial q_i} \pp - \frac{\partial V}{\partial {q}_i} -
{\bf e}_i^T CM^{-1} \pp .
\eea

The Fokker-Planck corresponding to (\ref{ksdknlwo33er}), which together with an initial distribution $f(\qq,\pp,0) = f_0(\qq,\pp)$ defines the family of time-evolving
pdfs $f(\qq,\pp; t)$, is
\beq
\frac{\partial f}{\partial t} + \sum_{i=1}^{n} \frac{\partial}{\partial q_i} \left(\alpha_i f\right) +
\sum_{i=1}^{n} \frac{\partial}{\partial p_i} \left(\gamma_i f\right) - \half \sum_{k=1}^{n}  \sum_{i,j=1}^{n} \frac{\partial^2}{\partial p_i \partial p_j}
\left(b_{ik} b_{kj}^T f\right) = 0
\label{kskdkf3333er}
\eeq
where $b_{ij} = {\bf e}_i^T B {\bf e}_j$ is the $i,j^{th}$ entry of $B$.

Note that for any mechanical system with inertia the diffusion is the same regardless of \Ito or Stratonovich interepretation, as
$$ \frac{\partial^2}{\partial p_i \partial p_j}
\left(b_{ik} b_{kj}^T f\right) = \frac{\partial}{\partial p_i}
\left(b_{ik} \frac{\partial}{\partial p_j} (b_{kj}^T f)\right)  =
b_{ik} b_{kj}^T \frac{\partial^2 f}{\partial p_i \partial p_j} \,. $$
That is, even though $B$ is configuration dependent, the structure of the FPE equations in the case
of mechanical systems with inertia places partial derivatives with respect to momenta in the diffusion terms,
and such partial derivatives pass through the configuration-dependent $B$ matrix. For this reason in
mechanical systems with inertia, it does not matter if \Ito or Stratonovich interpretations of SDEs are used.
This freedom allows the modeler to take the best of both worlds. But when approximations are made in modeling
the initial equations of motion, such as assuming that the inertia is negligible, then the above Hamiltonian formulation no longer applies and one must be very careful
as to the interpretation of the SDE as \Ito or Stratonovich.

\subsection{The Boltzmann Distribution}

The Boltzmann distribution is defined as
\beq
f_{\infty}(\qq,\pp) \,\doteq\, \frac{1}{Z} \exp \left(-\beta {H}({\bf p},{\bf q}) \right)
\label{maxboltdist}
\eeq
where $\beta \doteq 1/k_B T$ with $k_B$ denoting Boltzmann's constant and $T$ is temperature measured in degrees Kelvin.

The {\it partition function} is defined as
\begin{equation}
Z = \int_{{\bf q}} \int_{{\bf p}} \exp \left(-\beta \,{H}({\bf p},{\bf q}) \right) \,
d{\bf p} \, d{\bf q}.
\label{partition}
\end{equation}

The reason for using the subscript $\infty$ in defining (\ref{maxboltdist}) is
the following theorem.

\begin{theorem}
If $\qq \in \IR^n$ globally parameterizes the configuration manifold of a mechanical system, then the solution to the Fokker-Planck equation (\ref{kskdkf3333er}) will satisfy
$$ \lim_{t\rightarrow\infty} f(\pp,\qq,t) \,=\, f_{\infty}(\pp,\qq) $$
if and only if
\beq
C = \frac{\beta}{2} B B^T.
\label{eqdist332}
\eeq
\end{theorem}
\begin{proof}
We begin by noting that (\ref{kskdkf3333er}) can be simplified a bit. First,
$$ \frac{\partial}{\partial q_i} \left(\alpha_i f\right) \,=\,
\frac{\partial \alpha_i}{\partial q_i} f +  \alpha_i \frac{\partial f}{\partial q_i}
$$
and
$$ \frac{\partial}{\partial p_i} \left(\gamma_i f\right) \,=\,
\frac{\partial \gamma_i}{\partial p_i} f + \gamma_i \frac{\partial f}{\partial p_i} \,.
$$
It is not difficult to show that
$$  \sum_{i=1}^{n}\left\{
\frac{\partial \alpha_i}{\partial q_i} \,+\, \frac{\partial \gamma_i}{\partial p_i} \right\} \,=\,
{\rm tr}\left(C M^{-1}\right) \,. $$
Using this, and considering the equilibrium condition when ${\partial f}/{\partial t} = 0$ then reduces
(\ref{kskdkf3333er}) to
\beq {\rm tr}\left(C M^{-1}\right) f +
\sum_{i=1}^{n} \left\{\alpha_i \frac{\partial f}{\partial q_i} \,+\,
\gamma_i \frac{\partial f}{\partial p_i}\right\}- \half \sum_{i,j=1}^{n} \frac{\partial^2}{\partial p_i \partial p_j} \left((BB^T)_{ij} f\right) = 0 \,,
\label{kskdkf3333erdd}
\eeq
where the substitution
$$ (BB^T)_{ij} \,=\, \sum_{k=1}^{n} b_{ik} b_{kj}^T $$
has been made.

Note that
$$ \frac{\partial f_{\infty}}{\partial q_i} \,=\, - \beta \left(\frac{\partial V}{\partial {q}_i}
+ \frac{1}{2} \pp^T \frac{\partial M}{\partial q_i} \pp
\right) f_{\infty} \,,$$
$$ \frac{\partial f_{\infty}}{\partial p_i} \,=\, - \beta {\bf e}_i^T M^{-1} {\bf p} \, f_{\infty}
\,=\, - \beta \alpha_i \, f_{\infty} \,,$$
and hence significant cancelation results in
$$
\sum_{i=1}^{n} \left\{\alpha_i \frac{\partial f_{\infty}}{\partial q_i} \,+\,
\gamma_i \frac{\partial f_{\infty}}{\partial p_i}\right\} \,=\, \beta \bfalpha^T C \bfalpha \,.
$$
Moreover,
$$ \frac{\partial^2 f_{\infty}}{\partial p_i \partial p_i} \,=\,
\left(- \beta m_{ij}^{-1} + \beta^2 \alpha_i \alpha_j
\right) \, f_{\infty} $$
Substituting into (\ref{kskdkf3333erdd})
therefore gives
$$ {\rm tr}\left(C M^{-1}\right) \,+\, \beta \bfalpha^T C \bfalpha \,-\, \frac{\beta}{2} {\rm tr}\left(M^{-1} BB^T\right)
\,-\, \frac{\beta^2}{2} \bfalpha^T BB^T \bfalpha \,=\, 0. $$
This shows that $f_{\infty}(\pp,\qq)$ is in fact a solution to (\ref{kskdkf3333erdd}),
if (\ref{eqdist332}) holds. The necessary conditions for the above equality to hold boil down to the necessary
conditions for the two independent statements
$$
{\rm tr}\left((C \,-\, \frac{\beta}{2} BB^T)M^{-1}\right) \,=\, 0
$$
and
$$
\bfalpha^T (C \,-\, \frac{\beta}{2} BB^T) \bfalpha \,=\, 0
$$
to hold. The independence of these follows from the fact that some terms depend on $\bfalpha$ and others do not, and the main equality must hold for all values of $\bfalpha$.

The only way that both of the above can hold is if $C \,-\, \frac{\beta}{2} BB^T$ is skew symmetric.
But damping matrices, like stiffness and mass matrices, are symmetric, as is $BB^T$. Hence (\ref{eqdist332}) must hold.

Note that these necessary and sufficient conditions for the Boltzmann distribution to be the equilibrium solution is independent of any $\qq$-dependence of $B$ and $C$.
\end{proof}

%
%

\subsection{Marginal Densities and the Conundrum as Mass Becomes Zero}

Marginal densities of $f(\pp,\qq,t)$ can be defined as
$$ f(\pp,t) \,\doteq\, \int_{\qq} f(\pp,\qq,t) \,|\det M(\qq)|^{1/2} \,d\qq $$
and
\beq
f(\qq,t) \,\doteq\, |\det M(\qq)|^{-1/2} \int_{\pp} f(\pp,\qq,t) \,d\pp \,,
\label{marginalq}
\eeq
which is consistent with (\ref{fqdef}).

In the equilibrium case it is always possible to compute $f_{\infty}(\qq)$ in closed form
as
\beq
f_{\infty}(\qq) = \frac{1}{Z_c} e^{-\beta V(\qq)}
\label{confboltz}
\eeq
where
$$ Z_c \,=\, \int_{\qq} e^{-\beta V(\qq)} |\det M(\qq)|^{1/2} d\qq $$
is the configurational partition function. Then
$$ \int_{\qq} f_{\infty}(\qq) \, |\det M(\qq)|^{1/2} \, d\qq \,=\, 1\,. $$
In contrast, in general $f_{\infty}(\pp)$ can only be computed easily in closed form when $M(\qq) =M_0$ is constant. In this case
$$ f_{\infty}(\pp) = \frac{\beta^{d/2}}{(2\pi)^{d/2}|\det M_0|^{d/2}} \exp\left(-\frac{\beta}{2} \pp^T M_0^{-1} \pp\right) $$
is a Gaussian distribution in the momenta.

Though (\ref{confboltz}) degenerates as the inertia goes to zero, it does so gracefully since both $|\det M|^{1/2}$ and $Z_c$ approach zero in the same way as the system mass goes to zero. We can then use it as the baseline truth to compare approximations in which inertia is neglected. For example, consider the
spring-mass-damper with noise
$$ m \ddot{x} + c(x) \dot{x} + kx = b(x) n $$
where $c(x)$ and $b(x)$ are nonlinear functions satisfying the condition $2 c(x) = \beta b(x)^2$.
If $ndt = dw$, then as $m \rightarrow 0$, we have a conundrum unless $c=c_0$ and $b=b_0$ are constant. Namely, which of the following interpretations is correct ?
$$ d{x}_1 = - k\, c(x_1)^{-1} \,x_1 dt \,+\, 2\beta^{-1} b(x_1)^{-1} \, dw $$
or
$$ d{x}_2 = - k\, c(x_2)^{-1} \,x_2 dt \,+\, 2\beta^{-1} b(x_2)^{-1} \,\circledS\, dw \,? $$
It didn't matter when there was inertia, as both gave the same FPEs in the case,
but making the approximation that the mass is zero creates a situation
where a choice now must be made.

The answer can be informed by comparing the corresponding pdfs that solve the associated Fokker-Planck equations,
$f_1(x,t)$ and $f_2(x,t)$, with $f(x,t)$ in (\ref{marginalq}). Short of that, we can examine the behavior of the mean as a function of time, and the behavior of the equilibrium distributions as compared with (\ref{marginalq}).

The Fokker-Planck equations corresponding to the above SDEs are respectively
$$ \frac{\partial f_1}{\partial t} = k \frac{\partial}{\partial x}(c^{-1} x f_1) + 2\beta^{-2} \frac{\partial^2}{\partial x^2}(b^{-2} f_1) $$
and
$$ \frac{\partial f_2}{\partial t} = k \frac{\partial}{\partial x}(c^{-1} x f_2) + 2\beta^{-2} \frac{\partial}{\partial x}\left(b^{-1} \frac{\partial}{\partial x}(b^{-1} f_2)\right) \,. $$

Expanding and considering equilibrium conditions gives

$$ \Delta_1 = k \frac{\partial}{\partial x}(c^{-1} x f_1) + 2\beta^{-2} \frac{\partial}{\partial x}\left(-2b^{-3} \frac{\partial b}{\partial x} f_1 +
b^{-2} \frac{\partial f_1}{\partial x} \right) $$
and
$$ \Delta_2 = k \frac{\partial}{\partial x}(c^{-1} x f_2) + 2\beta^{-2} \frac{\partial}{\partial x}\left(-b^{-3}
\frac{\partial b}{\partial x} f_2 + b^{-2} \frac{\partial f_2}{\partial x} \right) $$
where an exact solution would give $\Delta_i =0$.

The exact configurational marginal from the Hamiltonian formulation is
$$ f_{\infty}(x) = \left(\frac{\beta k}{2\pi}\right)^{\half} e^{-\beta k x^2/2} \,, $$
and it has the property
$$ \frac{\partial f_{\infty}}{\partial x} \,=\, -\beta k x f_{\infty} \,. $$
Substituting into the above, and observing that
$$ k \frac{\partial}{\partial x}(c^{-1} x f_{\infty}) +  2\beta^{-2} \frac{\partial}{\partial x}\left(b^{-2} \frac{\partial f_{\infty}}{\partial x} \right) \,=\,0 $$
due to the relationship between $b$ and $c$, then
$$ \Delta_1 = 2\beta^{-2} \frac{\partial}{\partial x}\left(-b^{-3}
\frac{\partial b}{\partial x} f_{\infty}\right) \,=\, 2\Delta_2 \,. $$

This means that neither interpretation gives the true answer at equilibrium, but the magnitude of the discrepancy in the Stratonovich model is half that of the \Ito. For this reason, unless modeling systems in phase space, or if there are physical grounds for choosing a particular SDE (e.g., working backwards from Fick's Law), it the safest to consider diffusions with constant diffusion tensors, as will be the case throughout the remainder of this paper.

%
%
%
%
%

\section{Stochastic Systems on Unimodular Lie Groups}

A stochastic mechanical system that has a Lie group as its configuration space
can be studied in a coordinate-free way \cite{stochastic}.
These systems can be purely kinematic, or can have inertia.
Concrete examples are used here to illustrate, and then general theorems are provided to quantify the
rate of entropy production. Different connections between Lie groups and thermodynamics than what is presented here have been made in the literature \cite{Barbaresco1, Barbaresco2, Marle, Saxc}.

\subsection{Review of Unimodular Matrix Lie Groups with $SO(3)$ and $SE(2)$ as Examples}

The use of geometric (and particularly Lie-theoretic) methods in the control of mechanical systems and robots
has been studied extensively over the past half century \cite{15brocket73_applchap,murraylisastry_applchap,BulloLewisbook_applchap, Holm1, Holm2}.
The material and notation in this section summarizes more in-depth treatments in \cite{stochastic,harmonic}.

A matrix Lie group is a group with elements that are matrices, for which group multiplication is matrix multiplication, and for which the underlying space is an analytic manifold with the operations of group multiplication and inversion of elements being analytic also. Intuitively, matrix Lie groups are continuous
families of invertible matrices with structure that is preserved under multiplication and inversion. The dimension
of a matrix Lie group is the dimension of its manifold, not the dimension of the square matrices describing its elements.

For example, the group of rigid-body displacements in the Euclidean plane, $SE(2)$, can be described with
elements of the form
\beq
g(x,y,\theta) \,=\, \left(\begin{array}{ccc}
\cos \theta & -\sin\theta & x \\
\sin\theta & \cos\theta & y \\
0 & 0 & 1 \end{array}\right) \,.
\label{se2elem}
\eeq
The dimension is 3 because there are three free parameters $(x,y,\theta)$. This group is not compact as
$x$ and $y$ can take values on the real line.

The group of pure rotations in 3D can be described by rotation matrices
$$ SO(3) \,\doteq\, \{R \in \IR^{3\times 3} \,|\, RR^T = \II\,,\, \det R = +1\}\,. $$
$SO(3)$ is a compact 3-dimensional manifold. Again, the fact that the dimension of the matrices is also 3 is
coincidental.

A unimodular Lie group is defined by the property that a measure $dg$ can be constructed
such that the integral over the group has the property that
\beq
\int_G f(g)\, dg \,=\, \int_G f(g_0 \circ g)\, dg
\,=\, \int_G f(g \circ g_0)\, dg
\label{shiftinv}
\eeq
for any fixed $g_0 \in G$ and any function $f \in L^1(G)$.  It can also be shown that as a consequence of (\ref{shiftinv})
\beq
\int_G f(g)\, dg \,=\,  \int_G f(g^{-1})\, dg  \,.
\eeq
These properties are natural generalizations of those familiar to us for functions on Euclidean space.

As we are primarily concerned with
probability density functions for which
$$ \int_G f(g) dg = 1\,, $$
these clearly meet the condition of being in $L^1(G)$.

In the case of $SO(3)$ the bi-invariant measure expressed in terms of $Z-X-Z$ Euler angles
$(\alpha,\beta, \gamma)$ is $dR = \sin\beta d\alpha d\beta d\gamma$. In the case of $SE(2)$,
the bi-invariant measure is $dg = dx dy d\theta$.

The convolution of probability density functions on a unimodular Lie group is a natural operation,
and is defined as
\beq
(f_1 * f_2)(g) \,\doteq\, \int_{G} f_1(h) f_2(h^{-1} g) \, dh \,.
\label{convdef}
\eeq
The convolution of two probability density functions is again a probability density.

In addition to being natural spaces over which to integrate probability density functions,
natural concepts of directional derivatives of functions exist in the matrix Lie group setting.
This builds on the fact that associated with every matrix Lie group is a matrix Lie algebra.

In the case of $SO(3)$, the Lie algebra consists of $3\times 3$ skew-symmetric matrices of the form
\begin{equation}
{X} = \left(\begin{array}{ccc}
0 & -x_3 & x_2 \\
x_3 & 0 & -x_1 \\
-x_2 & x_1 & 0
\end{array} \right) = \sum_{i=1}^{3} x_i E_i.
\label{skew}
\end{equation}
The matrices $\{E_i\}$ form a basis for the set of $3\times 3$ skew-symmetric matrices.
The coefficients $\{x_i\}$ are all real. The notation relating the matrix $X$ and the vector
$\xx = [x_1,x_2,x_3]^T$ is \cite{murraylisastry_applchap,BulloLewisbook_applchap}
\beq
\xx \,=\, X^{\vee} \,\,\,{\rm and}\,\,\, X = \hat{\xx} \,.
\label{veedef}
\eeq
This is equivalent to identifying $E_i^{\vee}$ with ${\bf e}_i$.

For $SE(2)$ the basis elements are different, and are of the form
$$ E'_1 = \left(\begin{array}{ccc}
0 & 0 & 1 \\
0 & 0 & 0 \\
0 & 0 & 0 \end{array} \right); \hskip 0.2 true in
 E'_2 = \left(\begin{array}{ccc}
0 & 0 & 0 \\
0 & 0 & 1 \\
0 & 0 & 0 \end{array} \right); \hskip 0.2 true in
 E'_3 = \left(\begin{array}{ccc}
0 & -1 & 0 \\
1 & 0 & 0 \\
0 & 0 & 0 \end{array} \right) \,. $$
Every element of the Lie algebra associated with $SE(2)$ can be written as a linear combination of these, and
the notation (\ref{veedef}) is still used to identify these matrices with natural unit basis vectors
${\bf e}_i \in \IR^3$.
For example, $X' = x'_1 E'_1 + x'_2 E'_2 + x'_3  E'_3$. (Here the primes are used so as not to confuse the
Lie algebra elements for $SE(2)$ with those for $SO(3)$, but when working with a single Lie group and Lie algebra
the primes will be dropped, as in the discussion below which is for the generic case.)

For an arbitrary unimodular matrix Lie group, a natural concept of directional derivative is
\beq
(\tilde{X} f)(g) \,\doteq\, \left.\frac{d}{dt} f(g \, \exp(tX))\right|_{t=0} \,.
\label{deriv}
\eeq
Here the argument of the function $f$ is read as the product of $g$ and $\exp(tX)$, which are each in $G$, as is the product.
If $X = \sum_i x_i E_i$ for constants $\{x_i\}$, this derivative has the property
$$(\tilde{X} f)(g) = \sum_i x_i (\tilde{E}_i f)(g) \,. $$
Such derivatives appear in invariant statements of Fokker-Planck equations on unimodular Lie groups.
Moreover, these derivatives can be used together with integration to state results such as integration by parts
$$ \int_{G} f_1(g) (\tilde{E}_i f_2)(g) dg \,=\,  -\int_{G} f_2(g) (\tilde{E}_i f_1)(g) dg \,. $$
There are no surface terms because either group is infinite in its extent (and so the functions must decay to zero at the boundaries), or it is compact (in which case the functions must match values when arriving from different directions), or both for a group such as $SE(2)$.

\subsection{The Noisy Kinematic Cart}

The stochastic kinematic cart has been studied extensively in the robotics literature
\cite{Roussopoulos_applchap,thrun,wooramrobotica_applchap,zhou03_applchap,stochastic}. In this model (which is like a motor-driven wheelchair) the two wheels each have radii $r$, and the wheelbase (distance between wheels) is denoted as $L$.
The nonholonomic equations of motion are
\beq
\left(
\begin{array}{c}
\dot{x} \\ \\
\dot{y} \\ \\
\dot{\theta} \end{array}\right) \,=\,
\left(
\begin{array}{cc}
\frac{r}{2} \cos \theta & \frac{r}{2} \cos \theta \\ \\
\frac{r}{2} \sin \theta & \frac{r}{2} \sin \theta \\ \\
\frac{r}{L} & -\frac{r}{L} \end{array}\right)
\left(\begin{array}{c}
\dot{\phi}_1 \\ \\
\dot{\phi}_2 \end{array}\right) .
\label{sdecart2}
\eeq
When the wheel rates consist of a constant deterministic part and a stochastic part, then
\begin{eqnarray}
d\phi_1 \,&=&\, \omega\, dt + \sqrt{D}\, dw_1 \\
d\phi_2 \,&=&\, \omega\, dt + \sqrt{D}\, dw_2
\label{sdecart1}
\end{eqnarray}
and multiplying (\ref{sdecart2}) by $dt$ and substituting in (\ref{sdecart1}) results in an SDE.
This is an example where it does not matter whether the SDE is of It\^{o} or Stratonovich type, even though $B$ is not constant.
The corresponding Fokker-Planck equation for the probability density function $f(x,y,\theta;t)$ with respect to
measure $dx dy d\theta$ is \cite{zhou03_applchap}
$$
\frac{\partial f}{\partial t} = - r\omega \cos \theta  \frac{\partial f}{\partial x}
- r\omega \sin \theta  \frac{\partial f}{\partial y} +  $$
$$ \frac{D}{2} \left(\frac{r^2}{2} \cos^2 \theta
\frac{\partial^2 f}{\partial x^2} + \frac{r^2}{2} \sin 2\theta \frac{\partial^2 f}{\partial x \partial y}
+ \frac{r^2}{2} \sin^2 \theta \frac{\partial^2 f}{\partial y^2} + \frac{2r^2}{L^2} \frac{\partial^2 f}{\partial \theta^2}\right),
$$
which is subject to the initial conditions $f(x,y,\theta;0) = \delta(x-0)\delta(y-0)\delta(\theta-0).$

The coordinates $(x,y,\theta)$ that define the position and orientation of the cart relative to the world frame
are really parameterizing the group of rigid-body motions of the plane, $SE(2)$. Each element of this unimodular
Lie group can be described as homogeneous transformation matrices of the form in (\ref{se2elem})
in which case the group law is matrix multiplication.

Then (\ref{sdecart2}) can be written in coordinate-free notation as
\beq
\left(g^{-1} \frac{dg}{dt}\right)^{\vee} \,=\, A \, \dot{\bfphi} \,\,\,\,{\rm where}\,\,\,\,
A \,=\, \frac{r}{2}
\left(
\begin{array}{cc}
1 & 1 \\
0 & 0 \\
2/L & -2/L \end{array}\right) \,.
\label{cartcoordfree}
\eeq
Here the notation $\vee$ is used as in \cite{murraylisastry_applchap,stochastic,harmonic} in analogy with (\ref{veedef}), but for the case of $SE(2)$ rather than $SO(3)$.

The coordinate-free version of the above Fokker-Planck equation can be written compactly in terms of
these Lie derivatives as \cite{zhou03_applchap}
\beq
\frac{\partial f}{\partial t} = - r\omega \tilde{E}'_1 f + \frac{r^2 D}{4} (\tilde{E}'_1)^2 f + \frac{r^2 D}{L^2} (\tilde{E}'_3)^2 f
\label{compactfpex}
\eeq
with initial conditions $f(g;0) = \delta(g)$. The resulting time-evolving pdf is denoted as $f(g;t)$ with respect
to the natural bi-invariant integration measure for $SE(2)$, which is $dg = dx dy d\theta$.
Solutions for (\ref{compactfpex}) can be obtained in different regimes (small $Dt$ and large $Dt$) either
using Lie-group Gaussian distributions or Lie-group Fourier expansions, as in \cite{zhou03_applchap,wolfe}.
That is not the goal here. Instead, the purpose of this example is to provide a concrete case for the derivations
that follow regarding rate of entropy production.

It should be noted that degenerate diffusions on $SE(2)$ occur not only in this problem,
in models of the
visual cortex \cite{Mumford94,Williams97a,Williams97b,Zweck2004,Citti,DuitsFranken2010}.
Phase noise is a problem in coherent optical communication systems that has been identified in the literature.
\cite{com_bond89, com_foschini88b, com_foschini89, com_garrett90a}.
The Fokker-Planck equations describing phase noise have been developed and solved using various methods
\cite{com_garrett89b, com_garrett90a, com_zhang95,com_WangZhou}. Remarkably, these FPEs are the same kind as
those for the kinematic cart, inpainting, visual cortex modeling, etc. Moreover,
the natural extension of (\ref{cartcoordfree}) and (\ref{compactfpex}) to $SE(3)$ has found applications in modeling DNA (as reviewed in \cite{stochastic,harmonic}) and flexible steerable needles for robotic surgery \cite{park2005,wooramrobotica_applchap,Webster_needle}.

\subsection{Rotational Brownian Motion}

Starting with Perrin \cite{16Perrin3a}, various efforts at mathematical modeling of rotational Brownian motion has been undertaken over the past century
\cite{ 16furry,16favro,16Hubbard5a,16steele63a,16mcconnella}. These include both inertial and noninertial theories. A major application is in the spectroscopy of macromolecules \cite{Weber55,tao69}. Essentially the
same mathematics is applicable to modeling the time-evolving uncertaintly in mechanical gyroscopes
\cite{Willsky_applchap}.

Brownian motion on Riemannian manifolds and Lie groups also has been studied over a long period of time in the mathematics literature \cite{16ito,16ito1,16ito2,15mckean,stgr_Gangolli64,15duncana}, with the rotation group and three-sphere being two very popular objects \cite{16mckean,stgr_16Gorman,stgr_15Liao}. In addition to forcing by white noise, forcing by L\'{e}vy processes (white noise with jumps) has also been investigated \cite{stgr_liaobook}.

\subsubsection{Inertial Theory}

Euler's equation of motion for a rotating rigid body subjected to an external potential, noise, and damping can be written as
\beq
I_0 \,d{\bfomega} + \bfomega \times (I_0 \bfomega) \,dt = - (\tilde{\bf E} V)(R) dt - C_0  \bfomega dt
+ B_0 d{\bf w}
\label{iniertial1}
\eeq
where
$$ (\tilde{\bf E} V)(R) \,=\,\left(\begin{array}{c}
(\tilde{E}_1 V)(R) \\
(\tilde{E}_2 V)(R) \\
(\tilde{E}_3 V)(R) \end{array}\right)\,. $$
Here $\bfomega$ is the body-fixed description of angular velocity, which is related to a time-evolving rotation matrix (using the hat notation in (\ref{veedef})) as
\beq
\dot{R} = R \hat{\bfomega} \,.
\label{angmomdefdff}
\eeq
The moment of inertia matrix, $I_0$, damping matrix, $C_0$, and noise matrix
$B_0$ are all constant. Equations (\ref{iniertial1}) and (\ref{angmomdefdff}) define a stochastic process
evolving on the tangent bundle of $SO(3)$.

This can be re-written using angular momentum, $\bfpi = I_0 \bfomega$, as
\beq
d{\bfpi} \,=\, \bfpi \times (I_{0}^{-1} \bfpi) \,dt  - (\tilde{\bf E} V)(R) dt - C_0 I_{0}^{-1} \bfpi dt
+ B_0 d{\bf w} \,.
\label{iniertial2}
\eeq
Equations (\ref{iniertial2}) and
$$ \dot{R} = R \, \widehat{I_0^{-1} \bfpi} $$
define a stochastic process on the cotangent bundle of $SO(3)$.

Note that ${\bf p} \neq \bfpi$. To see this, expand angular velocity and kinetic energy in coordinates
as $\bfomega = J(\qq) \dot{\qq}$ and
$$ T = \half \dot{\qq}^T J(\qq)^T I_0 J(\qq) \dot{\qq}\,. $$
Consequently $M(\qq) = J(\qq)^T I_0 J(\qq)$ and ${\bf p} = M(\qq) \dot{\qq}$.
In contrast, $\bfpi = I_0 J(\qq) \dot{\qq}$. Therefore, in order to use the general results
from statistical mechanics, the interconversion
$$ {\bf p} \,=\, J(\qq)^T \bfpi $$
must be done. Moreover, $C_0$ in the above equations is not the same as $C$ in the Hamiltonian formulation.
A Rayleigh dissipation function will be of the form
$$ {\cal R} = \half \bfomega^T C_0 \bfomega = \half \dot{\qq}^T J(\qq)^T C_0 J(\qq) \dot{\qq}\,, $$
indicating that $C(\qq) = J(\qq)^T C_0 J(\qq)$.
Then converting (\ref{iniertial2}) to the Hamiltonian form, the viscous and noise terms become
$$ J(\qq)^T[-C_0 I_{0}^{-1} \bfpi dt + B_0 d{\bf w}] = -J(\qq)^T C_0 I_{0}^{-1} I_{0} J(\qq) d\qq + J(\qq)^T B_0 d{\bf w}\,. $$
If $C(\qq) = J(\qq)^T C_0 J(\qq)$ and $B(\qq) = J(\qq)^T B_0$, and if $J$ is invertible, then the condition
in (\ref{eqdist332}) then becomes completely equivalent to
\beq
C_0 = \frac{\beta}{2} B_0 B_0^T
\label{newbc}
\eeq
The structure of the $C_0$ matrix for a rigid body is a function of its shape.
For example, the viscous drag on an ellipsoid was characterized in \cite{jeffreys1922}.
Given $C_0$, it is possible to define $B_0 = C_{0}^{1/2}$.


When $I = \II$ and $V=0$ (\ref{iniertial2}) becomes
\beq
d{\bfpi} \,=\, -  \frac{\beta}{2} B_0 B_0^T \bfpi dt + B_0 d{\bf w} \,.
\label{iniertialnoV}
\eeq
This is an Ornstein-Uhlenbeck process, and the corresponding Fokker-Planck equation can be solved for
$f(\bfpi,t)$ in closed form as a time-varying Gaussian if the initial conditions are $f(\bfpi,0) = \delta(\bfpi)$.
The equilibrium solution is the Boltzmann distribution
$$ f_{\infty}(\bfpi) = c(\beta) \exp\left(-\frac{\beta}{2} \|\bfpi\|^2   \right)  $$
where $c(\beta)$ is the usual normalizing constant for a Gaussian distribution.

\subsubsection{Noninertial Theory}

When the inertia is negligible, as it is in the case of rotational Brownian motion of molecules, then (\ref{newbc}) and (\ref{iniertial1}) give
\beq\bfomega dt \,=\, B_1 d{\bf w} \,\,\,{\rm where}\,\,\, B_1 = \frac{2}{\beta} B_0^{-T} \,.
\,.
\label{iniertial3}
\eeq
This can be expressed in coordinates as a Stratonovich equation
\beq
\dot{\qq} = J^{-1}(\qq) B_1 {\bf w} \,,
\label{rotfp233d}
\eeq
or it can be kept in the invariant form (\ref{iniertial3}). The corresponding Fokker-Planck equation is of the form
$$ \frac{\partial f}{\partial t} = \sum_{i,j=1}^{3} D_{ij} \tilde{E}_i \tilde{E}_j f $$
where $D = B_1 B_1^T$ and each $\tilde{E}_i$ is as in (\ref{deriv})
with $X=E_i$.

The short-time solution to this equation subject to Dirac delta initial conditions is the Gaussian
in exponential coordinates. Hence, for short times, entropy and entropy rate can be computed in closed form using the results from the Euclidean case.

In noninertial theory a special case is isotropic diffusions. Let
$$ \nabla^2 \,\doteq\, \tilde{E}_1^2 + \tilde{E}_2^2 + \tilde{E}_3^2 \,. $$
An isotropic driftless diffusion on $SO(3)$ is one of the form
\begin{equation} \label{eqn:heat_SO3}
\frac{\partial f}{\partial t} = K \nabla^2 f \,.
\end{equation}
The heat kernel for $SO(3)$ is the solution to this subject to
the initial condition $f(R,0) = \delta (R)$.

Rotation matrices can be expressed in terms of the axis and angle of rotation using Euler's formula
$$ R(\theta,\nu,\lambda) \,=\, \exp[\theta \hat{\bf n}(\nu,\lambda)] \,=\, \II + \sin\theta \, \hat{\bf n}(\nu,\lambda) + (1-\cos \theta)\, [\hat{\bf n}(\nu,\lambda)]^2 $$
where $\hat{\bf n}(\nu,\lambda)$ is a skew symmetric matrix such that for arbitrary vector
${\bf v} \in \IR^3$ and vector cross product $\times$,
$$ \hat{\bf n}(\nu,\lambda) {\bf v} =
{\bf n}(\nu,\lambda) \times {\bf v}$$
and
$$ {\bf n}(\nu,\lambda) = \left(\begin{array}{c}
\sin \nu \cos \lambda \\
\sin \nu \sin \lambda \\
\cos \nu
\end{array}\right). $$
There are several different ways to choose the ranges of these coordinates to fully parameterize $SO(3)$.
One way is to view these coordinates as a solid ball of radius $\pi$ in which $\theta \in [0,\pi]$ serves
as the radius and $\nu \in [0,\pi]$ and $\lambda \in [0,2\pi)$, are the usual spherical angles.
Another way is to let $\theta \in [0,2\pi)$ and cut the range of one of the other variables in half.
For example $\nu \in [0,\pi/2]$ and $\lambda \in [0,2\pi)$ would restrict ${\bf n}$ to the upper hemisphere and
$\nu \in [0,\pi]$ and $\lambda \in [0,\pi)$ would be like the western hemisphere (if the initial datum is chosen appropriately). In these hemispherical boundary models, the great circle that bounds the hemisphere will be half open and half closed so as not to redundantly parameterize.

There are benefits to each of these parameterizations. For example, allowing the $[0,2\pi)$ range for
$\theta$ reflects that for fixed ${\bf n}$ rotations around a fixed axis bring back to the same location. That is
the `little group' of rotations around ${\bf n}$ isomorphic to $SO(2)$ is the `maximal torus' in $SO(3)$.
Likewise parameterizing the whole sphere has value. For this reason, the best of both worlds can be achieved
by double covering rotations by allowing both ranges to be expanded.
Moreover, each range $[0,2\pi)$ can be replaced with $[-\pi,\pi)$. Then when performing integration all that
needs to be done is to divide by $2$ afterwards.

Using these parameters, the integration measure $dR$ such that the volume of $SO(3)$ is normalized to $1$ is
\beq
dR \,=\, \frac{1}{4\pi^2} \sin^2(\theta/2) \sin{\nu}\,d\theta d\lambda d\nu \,.
\label{normhaar}
\eeq
When computing integrals,
\begin{eqnarray*}
\int_{SO(3)} f(R)\,dR \,&=&\, \frac{1}{2\pi^2}
\int_{\nu = 0}^{\pi} \int_{\lambda = 0}^{2\pi} \int_{\theta = 0}^{\pi}
f(R(\theta,\nu,\lambda)) \sin^2(\theta/2) \sin{\nu}\,d\theta d\lambda d\nu \\
\,&=&\, \frac{1}{2\pi^2}
\int_{\nu = 0}^{\pi/2} \int_{\lambda = 0}^{2\pi} \int_{\theta = -\pi}^{\pi}
f(R(\theta,\nu,\lambda)) \sin^2(\theta/2) \sin{\nu}\,d\theta d\lambda d\nu \\
\,&=&\, \frac{1}{2\pi^2}
\int_{\nu = 0}^{\pi} \int_{\lambda = 0}^{\pi} \int_{\theta = -\pi}^{\pi}
f(R(\theta,\nu,\lambda)) \sin^2(\theta/2) \sin{\nu}\,d\theta d\lambda d\nu \\
\end{eqnarray*}
Doubling the range gives
\begin{equation}
\int_{SO(3)} f(R)\,dR \,=\, \frac{1}{4\pi^2}
\int_{\nu = 0}^{\pi} \int_{\lambda = 0}^{2\pi} \int_{\theta = -\pi}^{\pi}
f(R(\theta,\nu,\lambda)) \sin^2(\theta/2) \sin{\nu}\,d\theta d\lambda d\nu
\label{volso3eul}
\end{equation}
and when $f(R(\theta,\nu,\lambda)) = f(\theta) = f(-\theta)$
\begin{eqnarray}
\int_{SO(3)} f(R)\,dR \,&=&\, \frac{2}{\pi}
\int_{\theta = 0}^{\pi}
f(\theta) \sin^2(\theta/2) \,d\theta  \nonumber \\
\,&=&\, \frac{1}{\pi}
\int_{\theta = -\pi}^{\pi}
f(\theta) \sin^2(\theta/2) \,d\theta \,.
\label{normpdf3fefef}
\end{eqnarray}

All normalizations are such that
$$ \int_{SO(3)} 1\,dR \,=\,1\,. $$

The Laplacian operator for $SO(3)$ in this axis-angle parameterization is \cite{harmonic,8varshalovich,9gelfand}
\begin{equation}
\nabla^2 = \frac{\partial^2}{\partial \theta^2}
+ \cot \theta/2 \frac{\partial}{\partial \theta} +
\frac{1}{4 \sin^2 \theta/2} \left(\frac{\partial^2}{\partial \nu^2}
+ \cot \nu \frac{\partial}{\partial \nu} +
\frac{1}{\sin^2 \nu} \frac{\partial^2}{\partial \lambda^2}\right)
\label{so3Laplaciandef}
\end{equation}

It can be shown that the isotropic solution does not depend on $\nu$ or $\lambda$, and so all that needs to be solved is \cite{14chetelat,19Lee05}
\begin{equation} \label{eqn:heat_eqn_SO3}
\frac{\partial f}{\partial t} = K \left( \frac{\partial^2 f}{\partial \theta^2} + \cot \theta / 2 \frac{\partial f}{\partial \theta} \right)
\end{equation}
subject to initial condition $f(R,0) = \delta(R)$.

A basis for all functions on $SO(3)$ that depend only on $\theta$ are the functions $\{\chi_l(\theta) \,|\,
l \in \mathbb{Z}_{\geq 0}\}$ where
$$ \chi_l(\theta) \,=\,  \frac{\sin (l + \frac{1}{2}) \theta}{\sin \frac{\theta}{2}} \,. $$
These are eigenfunctions of the Laplacian:
$$ \nabla^2 \chi \,=\, -l(l+1)  \chi\,. $$
Consequently, the Fourier series solution to the isotropic heat equation on $SO(3)$ is known to be
\begin{equation} \label{eqn:heat_soln_SO3_1}
\boxed{\,f(R(\theta,\nu,\lambda), t) \,=\, \sum_{l=0}^{\infty} (2l+1) \chi_l(\theta) \, e^{-l(l+1) Kt} \,=\,
\left(\sin \frac{\theta}{2}\right)^{-1} \sum_{l=0}^{\infty} (2l+1) \, {\sin \left((l + 1/2) \theta\right)} \, e^{-l(l+1) Kt} \,.\,}
\end{equation}
Note that
$$ \lim_{t \rightarrow \infty} f(R(\theta,\nu,\lambda), t) \,=\, 1 $$
and
$$ \int_{SO(3)} f(R,t)\,dR \,=\,1\,. $$
When $t=0$, the above becomes the Fourier series for the Dirac delta function
$$ f(R,0) \,=\, \delta(R) \,. $$

As with the case of the heat equation on the circle, an alternative solution exists, analogous to a folded Gaussian. Denote this solution as $\rho(\theta, t)$, and
\beq
f(R(\theta,\nu,\lambda), t) = e^{Kt/4} \left(\sin \frac{\theta}{2}\right)^{-1} \sum_{k \in \mathbb{Z}} \rho(\theta + 2\pi k, t) \,.
\label{fofrt}
\eeq

This can be derived in two steps. First let
$$ f(R(\theta,\nu,\lambda), t) \,=\, \left(\sin \frac{\theta}{2}\right)^{-1} h(\theta,t)\,. $$
Substituting in to (\ref{eqn:heat_eqn_SO3}) and simplifying then gives
\begin{equation}
\frac{\partial h}{\partial t} = K \left( \frac{\partial^2 h}{\partial \theta^2} + \frac{1}{4} h \right) \,.
\end{equation}
Then substituting $h(\theta,t) = e^{at} q(\theta,t)$ and simplifying gives that when $a =K/4$
\begin{equation}
\frac{\partial q}{\partial t} = K \frac{\partial^2 q}{\partial \theta^2} \,.
\label{1Dheatso3}
\end{equation}

The fundamental solution to this heat equation (the 1D heat kernel) is
$$ q(\theta,t) \,=\, \frac{1}{2\sqrt{\pi Kt}} e^{-\theta^2/4Kt}\,. $$
But this solution is not a satisfactory choice for $\rho(\theta,t)$ for several reasons.
First, $SO(3)$ is a 3D space, and so the normalization is not correct, as the normalization factor
should be proportional to $1/t^{3/2}$. But to arbitrarily change the
temporal dependence will cause the result to no longer be a solution to (\ref{1Dheatso3}). Secondly,
whereas division by $\sin(\theta/2)$ is ok in the definition of $\chi_l(\theta)$ because zeros are balanced
by zeros in the numerator, that is not the case here.

There is a way to solve both problems. That is by realizing that if $q(\theta,t)$ solves (\ref{1Dheatso3}), then
so too does $C \partial q/\partial\theta$ for arbitrary constant $C$. Conseqently, we take as a candidate solution \cite{14chetelat,19Lee05}

\begin{equation} \label{eqn:soln_heatlike_eqn}
\rho(\theta, t) = C \frac{1}{(\pi K t)^{3/2}} \theta e^{-\theta^2 / 4Kt} \,.
\end{equation}
The choice of normalizing factor $C$ is made such that (\ref{fofrt}) is a pdf. This is a constant, hence
independent of $t$. Choosing a relatively small value of $t$, the summation in (\ref{fofrt}) reduces to a single term and  (\ref{normpdf3fefef}) becomes

$$ \frac{1}{\pi}
C \frac{e^{Kt/4}}{(\pi K t)^{3/2}} \int_{\theta = -\pi}^{\pi} \theta e^{-\theta^2 / 4Kt} \sin (\theta/2) \,d\theta \,=\, 1\,. $$
Moreover, for small $t$ the integral over $[-\pi,\pi]$ can be replaced with an integral over the whole real line.
Consequently, since
$$ \int_{-\infty}^{\infty} \theta e^{-\theta^2 / 4Kt} \sin (\theta/2) \,d\theta \,=\, 2\sqrt{\pi} (Kt)^{3/2} e^{-Kt/4}  $$
then
$$ C = \pi^2/2 $$
and
\beq
\boxed{\,f(R(\theta,\nu,\lambda), t) =  \frac{\sqrt{\pi}}{2} \frac{e^{Kt/4}}{(K t)^{3/2}}
\left(\sin \frac{\theta}{2}\right)^{-1} \sum_{k \in \mathbb{Z}}
(\theta + 2\pi k) \,e^{-(\theta + 2\pi k)^2 / 4Kt} \,.\,}
\label{fofrta}
\eeq

Since $f(R,t) > 0$ and integrates to $1$ it is a valid pdf. What needs to be tested is
$$ \lim_{t\rightarrow\infty} f(R,t) \,=\, 1 $$
and
$$  \lim_{t\rightarrow 0} f(R,t) \,=\, \delta(R)\,. $$
If these conditions hold, then the solution is valid.

This provides a way to bound entropy in a way similar to the case of the circle.
The next section makes more general exact statements for all values of time.


\subsection{Rate of Entropy Production under Diffusion on Unimodular Lie Groups}

The entropy of a pdf on a Lie group is defined in (\ref{Gent}).
If $f(g,t)$ is a pdf that satisfies a diffusion equation, then some interesting properties of $S_f(t)$ that are independent of initial conditions result. For example, if $\dot{S}_{f} = dS_f/dt$, then differentiating
under the integral gives
\bea
\dot{S}_{f} &=& - \int_G \left\{\frac{\partial f}{\partial t} \log f + \frac{\partial f}{\partial t}
\right\} \,dg .
\eea
Moreover, since $f$ is a pdf,
$$ \int_{G} \frac{\partial f}{\partial t} \,dg =
\frac{d}{dt} \int_G f(g,t) \,dg = 0 \,. $$
and so the second term in braces in the expression for $\dot{S}_{f}$ integrates to zero.

Substituting the diffusion equation
\beq
\frac{\partial f}{\partial t} = \half \sum_{i,j=1}^{n} D_{ij} \tilde{E}_i \tilde{E}_j f
- \sum_{k=1}^{n} a_k \tilde{E}_k f
\label{difflie}
\eeq
into the expression for $\dot{S}_{f}$ gives \cite{stochastic,infomatr}
\bea
\dot{S} &=& - \int_G \left\{
\half \sum_{i,j=1}^{n} D_{ij} \tilde{E}_i \tilde{E}_j f
- \sum_{k=1}^{n} a_k \tilde{E}_k f
\right\} \log f \, dg \\
&=& - \half \sum_{i,j=1}^{n} D_{ij} \int_G (\tilde{E}_i \tilde{E}_j f) \log f \, dg
+ \sum_{k=1}^{n}  a_k \int_G (\tilde{E}_k f) \, \log f \, dg \\
&=& \half \sum_{i,j=1}^{n} D_{ij} \int_G  (\tilde{E}_j f) (\tilde{E}_i\log f)\, dg
- \sum_{k=1}^{n} a_k \int_G f \, (\tilde{E}_k \log f) \, dg \\
&=& \half \sum_{i,j=1}^{n} D_{ij} \int_G  \frac{1}{f} (\tilde{E}_j f) (\tilde{E}_i f)\, dg - \sum_{k=1}^{n} a_k \int_G \tilde{E}_k f \, dg \\
&=& \half \sum_{i,j=1}^{n} D_{ij} \int_G  \frac{1}{f} (\tilde{E}_j f) (\tilde{E}_i f)\, dg \\
&=& \half {\rm tr}[DF]
\eea
where $F = [F_{ij}]$ is the Lie-group version of the Fisher information matrix with entries
$$ F_{ij} \,\doteq\, \int_G  \frac{1}{f} (\tilde{E}_j f) (\tilde{E}_i f)\, dg \,. $$
Consequently,
$$ \half \lambda_{min}(D) {\rm tr}[F] \,\leq\, \dot{S} \,\leq\, \half \lambda_{max}(D) {\rm tr}[F]\,. $$
The above result is an extension of one presented in \cite{stochastic,infomatr}.

\subsection{The Generalized de Briujn Identity}

Here a theorem derived in \cite{stochastic,infomatr} is restated.

\begin{theorem} \label{th7.1}
Let the solution of the diffusion equation (\ref{difflie}) with constant ${\bf a} = [a_1,...,a_n]^T$ subject to the initial condition $f(g,0;D,{\bf a}) = \delta(g)$ be denoted as
$f_{D, {\bf a},t}(g) = f(g,t;D,{\bf a})$. Let $\alpha(g)$ be another differentiable pdf on the group. Then
\beq
\frac{d}{dt} S(\alpha* f_{D, {\bf a},t}) = \half {\rm tr}[D F(\alpha * f_{D, {\bf a},t})].
\label{debru323ii205}
\eeq
\end{theorem}
\begin{proof}
The solution of the diffusion equation
\beq
\frac{\partial \rho}{\partial t} = \half \sum_{i,j=1}^{n} D_{ij} \tilde{E}_i \tilde{E}_j \rho
- \sum_{k=1}^{n} a_k \tilde{E}_k \rho
\label{skk3k2s}
\eeq
subject to the initial conditions $\rho(g,0) = \alpha(g)$
 is $\rho(g,t) = (\alpha * f_{D,{\bf a},t)}(g)$. Then
computing the derivative of $S(\rho(g,t))$ with respect to time yields
\beq
\frac{d}{dt} S(\rho) = -\frac{d}{dt} \int_G \rho(g,t) \log \rho(g,t) \,dg = -\int_G \left\{
\frac{\partial \rho}{\partial t} \log \rho + \frac{\partial \rho}{\partial t} \right\} \,dg.
\label{skk322499}
\eeq
By substituting in (\ref{skk3k2s}), the partial with respect to time can be replaced with Lie derivatives.
But
$$ \int_{G} \tilde{E}_k \rho \,dg = \int_{G} \tilde{E}_i \tilde{E}_j \rho \,dg = 0 \,. $$
Consequently, the second term on the right side of (\ref{skk322499})
completely disappears. Using the integration-by-parts formula\footnote{There are no
surface terms. As with the circle and real line, each coordinate in the integral either wraps around
or goes to infinity.}
$$ \int_{G} f_1 \, \tilde{E}_k f_2 \,dg = - \int_{G} f_2 \, \tilde{E}_k f_1 \,dg, $$
with $f_1 = \log \rho$ and $f_2 = \rho$ then gives
\bea
 \frac{d}{dt} S(\alpha * f_{D,{\bf a},t}) &=& \half \sum_{i,j=1}^{n} D_{ij} \int_G  \frac{1}{\alpha * f_{D,{\bf a},t}} \tilde{E}_j (\alpha * f_{D,{\bf a},t}) \tilde{E}_i (\alpha * f_{D,{\bf a},t})\, dg \\
&=& \half \sum_{i,j=1}^{n} D_{ij} F_{ij}(\alpha * f_{D,{\bf a},t}) \, = \, \half {\rm tr}\left[D \, F(\alpha * f_{D,{\bf a},t})\right].
\eea
This means that
$$ S(\alpha * f_{D,{\bf a},t_2}) - S(\alpha * f_{D,{\bf a},t_1}) = \half \int_{t_1}^{t_2} {\rm tr}\left[D F(\alpha * f_{D,{\bf a},t})\right] dt. $$
\end{proof}

Whereas some inequalities of Information Theory generalize to the Lie group setting, as demonstrated above, others do not. For example under convolution on a Lie group (\ref{fishrffftight}) and (\ref{entpowine932}) do not hold in general. As for the CRB, there are versions for Lie groups applicable to small values of $Dt$ or with different concepts of covariance, but not in a way that is directly applicable to the scenarios formulated here.

%
%

\section{Conclusions}
Stochastic mechanical systems can describe individual representatives of a statistical mechanical ensemble (as in a rotor in rotational Brownian motion), or can be a stand-alone system subjected to noise, such as a kinematic cart robot. When these systems have mass, \Ito and Stratonovich SDEs lead to the same Fokker-Planck equation. Even in the case of inertia-free kinematic systems evolving on Lie groups, these can be the same. From the Fokker-Planck equation, the rate of entropy can be computed. For diffusion processes (with or without drift), the rate of entropy increase is related simply to the diffusion matrix and the Fisher information matrix. This result holds both in Euclidean spaces and on unimodular Lie groups (including cotangent bundle groups), which are a common configuration space for mechanical systems. As systems approach equilibrium, the entropy rate approaches zero. Two different ways to approach equilibrium are discussed: 1) when there is a restoring potential; 2) when the configuration space is bounded. By using the monotonicity and convexity properties of the logarithm function together with inequalities from Information Theory, computable bounds on entropy and entropy rate are established.

\end{document}